\documentclass{article}  
\usepackage{amssymb}
\usepackage{amsfonts}
\usepackage{amsbsy}
\usepackage{amsmath}

\usepackage{delarray,multirow,array}
\usepackage[latin1]{inputenc}
\usepackage{color}

\newtheorem{theorem}{Theorem}[section]
\newtheorem{corollary}{Corollary}[section]

\newtheorem{remark}{Remark}[section]
\newtheorem{example}{Example}[section]
\newtheorem{proposition}{Proposition}[section]
\def\F{{\mathbb{F}}}

\newenvironment{proof}{\textbf{Proof:}}{\hspace*{\fill}
\nolinebreak\hspace*{\fill}$\Box$\newline\vspace{1mm}}

\newcommand{\im}{\textnormal{im}}
\newcommand{\sgn}{\rm sgn}

\newcommand{\N}{{\mathbb N}}

\title{\LARGE \bf
A new class of superregular matrices and MDP convolutional codes
}

\author{\small P. Almeida, D. Napp and R. Pinto}

\begin{document}

\begin{abstract}
This paper deals with the problem of constructing superregular matrices that lead to MDP convolutional codes. These matrices are a type of lower block triangular Toeplitz matrices with the property that all the square submatrices that can possibly be nonsingular due to the lower block triangular structure are nonsingular. We present a new class of matrices that are superregular over a sufficiently large finite field $\F$. Such construction works for any given choice of characteristic of the field $\F$ and code parameters $(n,k,\delta)$ such that $(n-k) | \delta$. Finally, we discuss the size of $\F$ needed so that the proposed matrices are superregular.

\vspace{7mm}

\vspace{7mm}
\end{abstract}



\maketitle

\section{Introduction}

%

In recent years, renewed efforts have been made to further analyze the distance properties of convolutional codes \cite{Gluesing-Luerssen2006,Hutchinson2008b,Hutchinson2005,Hutchinson2008,Munoz2006,SmGlRo,Tomas2010PhDT,Tomas2009}. Convolutional codes with the maximum possible distance (for a given choice of parameters) are called maximum distance separable (MDS). However, for error control purposes it is also important to consider codes with large column distances.

The convolutional codes whose column distances increase as rapidly as possible for as long as possible are called maximum distance profile (MDP) codes. These codes are specially appealing for the performance of sequential decoding algorithms as they have the potential to have a maximum number of errors corrected per time interval. In \cite{RoSm} a non-constructive proof of the existence of such codes (for all transmission rates and all degrees) was given. However, the problem of how to construct MDP codes is far from being solved and very little is known about the minimum field size required for doing so. It turns out that this issue has been connected to the construction of a particular type of \emph{superregular} matrices. In \cite{Gluesing-Luerssen2006} a concrete construction of superregular matrices is given for all parameters $(n,k,\delta)$ although over a field with a large characteristic and size. In  \cite{Hutchinson2008} the size of the field needed to have a superregular matrix is studied. They provide a bound on this size and conjecture the existence of a much tighter bound based on examples and computer searches.

In this paper, we will address these issues and present a new class of matrices that are superregular over a sufficiently large finite field $\F$ of any characteristic. We also provide a bound on the required field size needed for such matrices to be superregular.

\section{Preliminaries: MDP convolutional codes and superregular matrices}

In this section, we recall basic material from the theory of convolutional codes that is relevant to the presented work and link it to the notion of superregular matrices.

In this paper, we consider convolutional codes constituted by codewords having finite support: Let $\mathbb F$ be a finite field and $\mathbb F[z]$ the ring of polynomials with coefficients in $\mathbb F$. A \textbf{convolutional code} ${\cal C}$ of rate $k/n$ is a $\mathbb{F}[z]$-submodule of $\mathbb{F}[z]^n$ of rank $k$ of the form
\[{\mathcal C}  =  \im_{\mathbb F[z]}G(z) = \{G(z)u(z):\, u(z) \in \mathbb F^{k}[z]\},\]
where $G(z)\in \mathbb{F}[z]^{n \times k}$ is a right-invertible matrix over $\mathbb F[z]$, i.e., there exists a matrix, called the parity check matrix, $H(z)\in\mathbb{F}[z]^{(n-k) \times n} $ such that
\begin{equation}\label{eq:conv_code_parity_check}
     \im_{\mathbb F[z]} G(z)=\ker_{\mathbb F[z]} H(z)=\{ v(z) \in \mathbb{F}[z]^n : \ H(z)v(z)=0 \}.
\end{equation}
The \emph{degree} of $\cal C$, denoted by $\delta$, is defined as the maximum degree of the full size minors of $G(z)$. Notice that we can also choose $H(z)$ to be left invertible over $\mathbb F[z]$, and in this case $\delta$ will also be equal to the maximum degree of the full size minors of $H(z)$. A convolutional code of rate $k/n$ and degree $\delta$ is called an $(n,k, \delta)$ convolutional code.

The most important property of a code is its distance, defined as follows: The weight of a polynomial vector $v(z)=\sum_{i \in \mathbb N} v_i z^i\in \mathbb F[z]^n$ is given by ${\rm wt}(v)=\sum_{i\in \mathbb N} {\rm wt}(v_i)$, where ${\rm wt}(v_i)$ is the number of nonzero elements of $v_i$. The \textbf{distance} of a convolutional code $\cal C$ is defined as
\[{\rm d}({\cal C})={\rm min}\{{\rm wt}(v(z)) \, | \, v(z) \in {\cal C}, v(z) \neq 0 \}.\]
If ${\cal C}=\ker_{\mathbb F[z]}  H(z)$, where $H(z)=\displaystyle \sum_{i=0}^{\nu} H_i z^i$, for some $\nu \in \mathbb N$, then the $j$-th \textbf{column distance} of $\cal C$ is defined as
\begin{eqnarray*}
d_j^c &=& \min\{ wt(v_{[0,j]})=wt(v_0+ v_1z + \dots + v_j z^j) : \ v(z)=\sum_{i \in \mathbb N} v_i z^i\in {\cal C} \mbox{ and } \ v_0 \neq 0  \}\\
      &=&  \min\{ wt(\vec{v}_j) : \ \vec{v}_j =[v_0 \dots v_j]\in \F^{(j+1)n}, \ {\cal H}(H_0,\dots,H_j) \vec{v}_j^\top=0,  \ v(z)=\sum_{i \in \mathbb N} v_i z^i\in {\cal C} \mbox{ and }\ v_0\neq 0  \}.
\end{eqnarray*}
where
\begin{equation}\label{eq:08}
{\cal H}(H_0,\dots,H_j)=\left(
  \begin{array}{ccccc}
    H_0 & 0 & 0 & \cdots & 0 \\
    H_1 & H_0 & 0 & \cdots & 0 \\
    H_2 & H_1 & H_0 & \cdots & 0 \\
    \vdots & \vdots & \vdots & \vdots & \vdots \\
    H_j & H_{j-1} & \cdots & \cdots & H_0 \\
  \end{array}
\right) \in \mathbb F^{(j+1)(n-k)\times (j+1)n},
\end{equation}
and $H_j=0$ for $j>\nu$.

In this paper we focus on this important notion of column distance. This notion is closely related to the notion of optimum distance profile (ODP), see \cite[pp.112]{johannessonZ99}. The following results about column distances are proved in \cite{Gluesing-Luerssen2006}.

\begin{proposition}
Let $\cal C$ be an $(n,k, \delta)$ convolutional code and $L=\lfloor \dfrac{\delta}{k} \rfloor + \lfloor \dfrac{\delta}{n-k} \rfloor$. Then
\begin{itemize}
\item[i)] $d^c_j({\cal C}) \leq (j+1)(n-k)+1, \; \forall j \in \mathbb N_0$;
\item[ii)] if there exists $j \in \mathbb{N}_0$ such that $d^c_j({\cal C}) = (j+1)(n-k)+1$, then $d^c_i({\cal C}) = (i+1)(n-k)+1$, for $i \leq j$ and $j \leq L$.
\end{itemize}
\end{proposition}
A convolutional code $\cal C$ is called \textbf{maximum distance profile} (\textbf{MDP}) if its column distances achieve the maximum possible values (for a given choice of parameters), i.e., if $\cal C$ has rate $k/n$ and degree $\delta$, then $d_L^{c}({\cal C})=(L+1)(n-k)+1$, for $L=\lfloor \delta/k \rfloor + \lfloor \delta/(n-k) \rfloor$ and so $d_j^{c}({\cal C})=(j+1)(n-k)+1$, for $j \leq L$. In order to characterize MDP codes we need to introduce the notion of superregular matrices.

Let $A=[\mu_{ij}]$ be a square matrix of order $m$ over $\mathbb{F}$ and $S_m$ the symmetric group of order $m$. Recall that the determinant of $A$ is given by
\begin{equation}\label{deter}
|A|=\sum_{\sigma\in S_m}(-1)^{\sgn(\sigma)}\mu_{1\sigma(1)}\cdots \mu_{m\sigma(m)}.
\end{equation}
A {\em trivial term} of the determinant is a term of (\ref{deter}), $\mu_{1\sigma(1)}\cdots \mu_{m\sigma(m)}$, such that exists $1\leq i\leq m$ with $\mu_{i\sigma(i)}=0$. If $A$ is a square submatrix of a matrix $B$, with entries in $\mathbb{F}$, and all the terms of the determinant of $A$ are trivial we say that $|A|$ is a \textbf{trivial minor} of $B$. We say that $B$ is \textbf{superregular} if all its non-trivial minors are different from zero.\\

It is important to remark here that there exist several related, but different, notions of superregular matrices in the literature. Unfortunately, all these notions are only particular cases of the more general definition given above. Frequently, see for instance \cite{Roth1989}, a superregular matrix is defined to be a matrix for which every square submatrix is nonsingular. Obviously all the entries of these matrices must be nonzero. Also, in \cite{Aid1986,MaSl,Roth1985}, several examples of triangular matrices were constructed in such a way that all submatrices inside this triangular configuration were nonsingular. However, all these notions do not apply to our case as they do not consider submatrices that contain zeros. The more recent contributions \cite{Gluesing-Luerssen2006,Hutchinson2008b,Hutchinson2008,Tomas2009,Tomas2010PhDT} consider the same notion of superregularity as us, but defined only for lower triangular matrices.

Next theorem shows how MDP $(n, k, \delta)$ convolutional codes with $(n-k)| \delta$ can be characterized by superregular matrices (see \cite[Theorem 3.1]{Gluesing-Luerssen2006}).

\begin{theorem}\label{th:MDP_char}
Let ${\cal C}$ be an $(n,k, \delta)$ convolutional code such that $(n-k)| \delta $ and represented as ${\cal C}=\ker_{\mathbb F[z]}  [A(z) \ B(z) ]$  where $A(z)=\displaystyle \sum_{i=0}^{\nu} A_i z^i\in \mathbb F[z]^{(n-k)\times (n-k)}$, $B(z)=\displaystyle \sum_{i=0}^{\nu} B_i z^i\in \mathbb F[z]^{(n-k)\times k}$, $\nu= \frac{\delta}{(n-k)}$. We can assume without lost of generality that $A_0=I_{n-k}$. Furthermore, let
$$
A(z)^{-1}B(z) = \displaystyle \sum_{i=0}^{\infty} {\bar H}_i z^i\in \mathbb F((z))^{(n-k)\times k}
$$
be the Laurent expansion of $A(z)^{-1}B(z) $ over the field $\mathbb F((z))$ of Laurent series. Define $L=\lfloor \delta/k \rfloor +  \delta/(n-k)$ and
$$
 \widehat{{\bar H}}= [I_{(L+1)(n-k)} \ {\cal {\bar H}}({\bar H}_0,\dots,{\bar H}_L)] \mbox{ where}
$$
\begin{equation}\label{eq:04}
{\cal {\bar H}}({\bar H}_0,\dots,{\bar H}_L)=\left(
  \begin{array}{ccccc}
    {\bar H}_0 & 0 & 0 & \cdots & 0 \\
    {\bar H}_1 & {\bar H}_0 & 0 & \cdots & 0 \\
    {\bar H}_2 & {\bar H}_1 & {\bar H}_0 & \cdots & 0 \\
    \vdots & \vdots & \vdots & \vdots & \vdots \\
    {\bar H}_L & {\bar H}_{L-1} & \cdots & \cdots & {\bar H}_0 \\
  \end{array}
\right) \in \mathbb F^{(L+1)(n-k)\times (L+1)k}.
\end{equation}

The following are equivalent:
\begin{enumerate}
  \item $\cal C$ is MDP.
  \item ${\cal {\bar H}}({\bar H}_0,\dots,{\bar H}_L)$ is superregular.
\end{enumerate}
\end{theorem}
Hence, the problem of constructing an MDP convolutional code relies on the problem of constructing superregular lower block triangular Toeplitz  matrices of the form  (\ref{eq:04}). This problem is addressed in the next section.

For the case where $(n-k)\nmid \delta$, similar results were obtained using different methods from systems theory, see \cite{Hutchinson2008b,Hutchinson2005,Hutchinson2008} for more details. We will not consider this case in this paper.

%

\section{A new class of MDP codes and superregular matrices}

In this section, we introduce a new class of matrices of the form (\ref{eq:04}) and show that they are superregular matrices over a sufficiently large field $\F$. We conclude the section by providing a lower bound on the field size of $\F$ that ensures the superregularity of the proposed matrices. First, we recall previous contributions on superregular matrices.\\

It is a common practice in building the matrix ${\cal {\bar H}}({\bar H}_0,\dots,{\bar H}_L)$ of Theorem \ref{th:MDP_char} to first construct a large lower triangular superregular matrix in such a way that it contains the lower block triangular Toeplitz matrix ${\cal {\bar H}}({\bar H}_0,\dots,{\bar H}_L)$ as a submatrix.  In \cite{Gluesing-Luerssen2006}, it was shown that for every positive integer $r$ exists a prime $p=p(r)$ such that

\begin{equation}\label{eq_06}
S_r=\left(
  \begin{array}{ccccc}
    \left(
      \begin{array}{c}
        r-1 \\
        0 \\
      \end{array}
    \right)
     & 0 & 0 & \cdots & 0 \\
    \left(
      \begin{array}{c}
        r-1 \\
        1 \\
      \end{array}
    \right) & \left(
      \begin{array}{c}
        r-1 \\
        0 \\
      \end{array}
    \right) & 0 & \cdots & 0 \\
    \left(
      \begin{array}{c}
        r-1 \\
        2 \\
      \end{array}
    \right) & \left(
      \begin{array}{c}
        r-1 \\
        1 \\
      \end{array}
    \right) & \left(
      \begin{array}{c}
        r-1 \\
        0 \\
      \end{array}
    \right) & \cdots & 0 \\
    \vdots & \vdots & \vdots & \vdots & \vdots \\
    \left(
      \begin{array}{c}
        r-1 \\
        r-1 \\
      \end{array}
    \right) & \left(
      \begin{array}{c}
        r-1 \\
        r-2 \\
      \end{array}
    \right) & \cdots & \cdots & \left(
      \begin{array}{c}
        r-1 \\
        0 \\
      \end{array}
    \right) \\
  \end{array}
\right)
\end{equation}
is superregular over $\mathbb F_{p}$. Moreover, the authors proposed the first rough bound on the size of a field $\F$ for a lower triangular Toeplitz matrix $A$ to be superregular over $\F$. Namely if we consider $c$ to be the largest magnitude among the entries of $A$ and if $|\F|>c^r r^{r/2}$, then there exists a superregular lower triangular Toeplitz matrix $A\in\F^{r\times r}$. Later, in \cite{Hutchinson2008}, the following more refined bound was presented:
If $| \F| > B_r$ then there exists a superregular lower triangular Toeplitz matrix $A\in\F^{r \times r}$, where \begin{equation}\label{eq:03}
    B_r=\frac{1}{2}\left( \frac{1}{r} \left(\begin{array}{c}
                              2(r-1) \\
                              r-1
                            \end{array}
    \right)+\left(\begin{array}{c}
                              r-1 \\
                              \lfloor \frac{r-1}{2}\rfloor
                            \end{array}
    \right) \right).
\end{equation}

Moreover, based on examples and computer searches, it was conjectured in \cite{Gluesing-Luerssen2006,Hutchinson2008} that for $\ell\geq 5$ there exists a superregular lower triangular Toeplitz matrix of order $\ell$ over the field $\mathbb F_{2^{\ell-2}}$. If true, it would considerably improve the bound given above. This remains an open problem. \\

We propose a new type of superregular matrices with the form of (\ref{eq:04}). Of course, this will bring about a new class of MDP codes. Let $(n,k,\delta)$ be given such that $(n-k) | \delta$. Let $M=\max\{n-k,k \}$ and $L=\lfloor \delta/k \rfloor + \delta/(n-k)$. Let $\alpha$ be a primitive element of a finite field $\F=\mathbb F_{p^{\mathbb N}}$ and define
\begin{eqnarray}\label{eq_01}
\nonumber && [T_0 | \ T_1 \ | \dots \ |  T_L]= \\
&&\!\!\!\!\!\!\!\!\!\!= \!\! \left[ \!\!
  \begin{array}{cccc|ccc|c|ccc}
    \alpha^{2^0} & \alpha^{2^1} &\cdots & \alpha^{2^{M-1}} & \alpha^{2^{M}} & \cdots & \alpha^{2^{2M-1}} &  &\alpha^{2^{ML}} &  \cdots & \alpha^{2^{M(L+1)-1}} \\
    \alpha^{2^1} & \alpha^{2^2} &\cdots & \alpha^{2^{M}} & \alpha^{2^{M+1}} & \cdots & \alpha^{2^{2M}} & &\alpha^{2^{ML+1}} & \cdots & \alpha^{2^{M(L+1)}} \\
    \alpha^{2^2} & \alpha^{2^3} &\cdots & \alpha^{2^{M+1}} & \alpha^{2^{M+2}}& \cdots & \alpha^{2^{2M+1}}& \cdots & \alpha^{2^{ML+2}}  &  \cdots & \alpha^{2^{M(L+1)+1}} \\
    \vdots & \vdots &\ddots & \vdots & \vdots & \ddots & \vdots &  & \vdots & \ddots & \vdots \\
     \alpha^{2^{M-1}} & \alpha^{2^{M}} &\cdots & \alpha^{2^{2M-2}} & \alpha^{2^{2M-1}} & \cdots & \alpha^{2^{3M-2}} & &\alpha^{2^{M(L+1)-1}}  & \cdots & \alpha^{2^{M(L + 2) -2}} \\
  \end{array}
\!\! \right].
\end{eqnarray}

Define also, ${\cal T}( T_0,T_1, \dots  , T_{L})\in \F^{(L+1)M\times (L+1)M}$ by
\begin{equation}\label{eq:07}
{\cal T}(T_0,\dots,T_L)=\left(
  \begin{array}{ccccc}
    T_0 & 0 & 0 & \cdots & 0 \\
    T_1 & T_0 & 0 & \cdots & 0 \\
    T_2 & T_1 & T_0 & \cdots & 0 \\
    \vdots & \vdots & \vdots & \ddots & \vdots \\
    T_L & T_{L-1} & \cdots & \cdots & T_0 \\
  \end{array}
\right).
\end{equation}


We are going to prove that if $N$ is sufficiently large then ${\cal T}( T_0,T_1, \dots  , T_{L})$ is superregular. First, we need the following well known result.

\begin{theorem} [\cite{hungerford74}] \label{th:classic_finite_fields}
Let $\mathbb{F}$ be a finite field with $p^N$ elements. Let $\alpha$ be a primitive element of $\mathbb F$ and $\rho(z)$ be the minimal polynomial of $\alpha$ (i.e., $\mathbb F= \mathbb F_p[z]/(\rho(z))$ and deg $\rho(z)=N$). If $f(z) \in \mathbb F_p[z]$ with $f(\alpha)=0$ then $\rho(z)\mid f(z)$.
\end{theorem}

\begin{theorem}\label{th:mainFieldSize}
Let $\alpha$ be a primitive element of a finite field $\F$ of characteristic $p$, $\rho(z)$ be the minimal polynomial of $\alpha$ and consider ${\cal T}( T_0,T_1, \dots  , T_{L})\in \F^{(L+1)M\times (L+1)M}$. If $| \mathbb F | \geq p^{(2^{M(L+2)-1})}$ then the matrix ${\cal T}( T_0,T_1, \dots  , T_{L})$ is superregular (over $\F$).
\end{theorem}

\begin{proof}
Let $[t_{L1} \ \cdots \ t_{LM} | \cdots | t_{11} \ \cdots \ t_{1M} | t_{01} \ \cdots \ t_{0M}]$ denote the columns of ${\cal T}( T_0, \dots  , T_{L})$ and define $\overline{\cal T}( T_0, \dots  , T_{L})=[t_{01} \ \cdots \ t_{0M} |  t_{11} \ \cdots \ t_{1M} | \cdots | t_{L1} \ \cdots \ t_{LM}]$, i.e., set
\begin{eqnarray*}
&&\overline{\cal T}( T_0, \dots  , T_{L})= \\
&& \!\!\!\!\!\!\left[
  \begin{array}{ccc|c|ccc|ccc}
    0 & \cdots & 0 & \cdots & 0 & \cdots & 0 & \alpha^{2^0} & \cdots & \alpha^{2^{M-1}} \\
    0 & \cdots & 0 & \cdots & 0 & \cdots & 0 & \alpha^{2^1} & \cdots & \alpha^{2^{M}} \\
    \vdots & \ddots & \vdots & \ddots & \vdots & \ddots & \vdots & \vdots & \ddots & \vdots \\
    0 & \cdots & 0 & \cdots & 0 & \cdots & 0 & \alpha^{2^{M-1}} & \cdots & \alpha^{2^{2M-2}} \\
    \hline
    0 & \cdots & 0 & \cdots & \alpha^{2^0} & \cdots & \alpha^{2^{M-1}} & \alpha^{2^{M}} & \cdots & \alpha^{2^{2M-1}} \\
    0 & \cdots & 0 & \cdots &\alpha^{2^1} & \cdots & \alpha^{2^{k}} & \alpha^{2^{M+1}} & \cdots & \alpha^{2^{2M}} \\
    \vdots & \ddots & \vdots & \ddots & \vdots & \ddots & \vdots & \vdots & \ddots & \vdots \\
    0 & \cdots & 0 & \cdots &\alpha^{2^{M-1}} & \cdots & \alpha^{2^{2M-2}} & \alpha^{2^{2M-1}} & \cdots & \alpha^{2^{3M-2}} \\
    \hline
   \vdots & \ddots & \vdots & \ddots & \vdots & \ddots & \vdots & \vdots & \ddots & \vdots \\
   \hline
   \alpha^{2^0} & \cdots & \alpha^{2^{M-1}} & \cdots &\alpha^{2^{M(L-1)}} & \cdots & \alpha^{2^{ML-1}} & \alpha^{2^{ML}} & \cdots & \alpha^{2^{M(L+1) -1}}\\
   \alpha^{2^1} & \cdots & \alpha^{2^{M}} & \cdots &  \alpha^{2^{M(L-1)+1}} & \cdots & \alpha^{2^{ML}} & \alpha^{2^{ML+1}} & \cdots & \alpha^{2^{ML+k}}\\
   \vdots & \ddots & \vdots & \ddots & \vdots & \ddots & \vdots & \vdots & \ddots & \vdots \\
   \alpha^{2^{M-1}} & \cdots & \alpha^{2^{2M-2}} & \cdots & \alpha^{2^{ML-1}} & \cdots & \alpha^{2^{M(L+1)-2}}& \alpha^{2^{M(L+1)-1}} & \cdots & \alpha^{2^{M(L+2)-2}}
  \end{array}
\right].
\end{eqnarray*}

Next, we show that $\overline{{\cal T}}( T_0, \dots  ,T_{L})$ is superregular. Obviously, this readily implies that ${\cal T}( T_0, \dots  , T_{L})$ is superregular as well.

Let $A=[\mu_{ij}]$ be a square submatrix of $\overline{{\cal T}}( T_0, \dots  , T_{L})$ of order $m\leq M(L+1)$. Note that from the particular structure of the proposed matrix $\overline{{\cal T}}( T_0, \dots  , T_{L})$ it follows that

\begin{equation}\label{prop}
\mu^2_{ij'}\leq \mu_{ij}\hspace{5mm}\mbox{and}\hspace{5mm}\mu^2_{i'j}\leq \mu_{ij}\hspace{5mm}\mbox{if}\hspace{5mm}i'<i\hspace{5mm}\mbox{and}\hspace{5mm}j'< j.
\end{equation}

Consider $m>1$ otherwise the proof is trivial.
Write $A$ as a block matrix in the following form
\begin{equation}\label{Mgeral}
A=\left[\begin{tabular}{cccc|c}
& & $O_1$ & &  \multirow{5}{*}{$A_0$}\\\cline{1 - 4}
& $O_2$ &  & $\vline\hfill$\multirow{4}{*}{$A_1$} & \\\cline{1 -3}
& $\vdots$ &  & $\vline\hfill$ & \\\cline{1 - 1}
$O_h$ & $\vline\hfill$ & $\cdots$ &  $\vline\hfill$ & \\\cline{1 - 1}
$A_{h}$ & $\vline\hfill$ &  & $\vline\hfill$ &
\end{tabular}
\right],
\end{equation}
where, for each $1\leq i\leq h$, $O_i$ is a null matrix with $l_i$ columns and, for each $0\leq j\leq h$, $A_j$ is a matrix with $k_j$ rows and no entry equal to zero. We have $l_1>\dots>l_h$ and $m=k_0>k_1>\dots >k_h.$
The minor $|A|$ being nontrivial implies $k_i\geq l_i$ for any $1\leq i\leq h$.
\\

Notice that each term of the determinant of $A$ given by (\ref{deter}) is zero or a power of $\alpha$. We will prove that $\overline{{\cal T}}( T_0, \dots  , T_{L})$ is superregular by showing that if there are nontrivial terms in the determinant of $A$, then there exists a \emph{unique} term with highest exponent and thus $|A|\neq 0$. Let

\[
\beta =\max_{b\in\mathbb{N}}\{b: \alpha^b= \mu_{1\sigma(1)}\cdots \mu_{m\sigma(m)} , \mbox{ for some } \sigma\in S_m\} .
\]

Thus, it is enough to show that there is a unique $\overline{\sigma}\in S_m$ such that
\begin{equation}\label{eq_00}
\alpha^{\beta}=\mu_{m \overline{\sigma}(m)}\mu_{m-1 \overline{\sigma}(m-1)}\cdots \mu_{1 \overline{\sigma}(1)}.
\end{equation}
To this end, we first prove that $\overline{\sigma}(m)$ is \emph{\textbf{uniquely}} determined by the following rule:

\underline{\textbf{Case 1:}} If $h=0$ or $l_i<k_i$ for $i=1,\dots, h$, then $\overline{\sigma}(m)=m$.
\\

\emph{Proof Case 1:} Take $\widehat{\sigma}\in S_m$ with
$$
\mu_{m \widehat{\sigma}(m)}\mu_{m-1 \widehat{\sigma}(m-1)}\cdots \mu_{1 \widehat{\sigma}(1)}\neq 0,
$$
with $\mu_{m \widehat{\sigma}(m)} \neq \mu_{m m}$. Let $\widehat{\beta}\in \mathbb{N}$, such that
\[\alpha^{\widehat{\beta}}= \mu_{m \widehat{\sigma}(m)}\mu_{m-1 \widehat{\sigma}(m-1)}\cdots \mu_{1 \widehat{\sigma}(1)} .
\]

Since $h=0$ or $l_j<k_j$ for all $j=1,\dots, h$ then $\mu_{(m-i)i}\neq 0$ for any $i=1, \dots, m-1$ (if for some $i$, $\mu_{(m-i)i}=0$ then there exists $j\in\{1, \dots, h\}$ such that $l_j\geq i$ and $k_j\leq i$, a contradiction).

Construct $\widetilde{\sigma}\in S_m$ recursively, as follows:

\begin{enumerate}
\item Define $\delta_1=m$ and while $\mu_{\widehat{\sigma}^{-1}(\delta_i)\widehat{\sigma}(m)}= 0$, let
\[\delta_{i+1}=\widehat{\sigma}\left (\max_{j\geq m-\widehat{\sigma}^{-1}(\delta_i)}{\widehat{\sigma}^{-1}(j)}\right ).\]
Let $i_0$ be the first integer such that $\mu_{\widehat{\sigma}^{-1}(\delta_{i_0})\widehat{\sigma}(m)}\neq 0$;
\item $\widetilde{\sigma}(m)=m$ and $\widetilde{\sigma}(\widehat{\sigma}^{-1}(\delta_{m_0}))=\widehat{\sigma}(m)$;
\item For $1\leq i\leq i_0$,  $\widetilde{\sigma}(\widehat{\sigma}^{-1}(\delta_i))=\delta_{i+1}$;
\item For $i\not\in I=\{\widehat{\sigma}^{-1}(\delta_i)\; \mid\; i=1, \dots,i_0\}$,  $\widetilde{\sigma}(i)=\widehat{\sigma}(i)$.
\end{enumerate}

Clearly, $\widehat{\sigma}^{-1}(\delta_i)>\widehat{\sigma}^{-1}(\delta_{i-1})$ and by (\ref{prop}),
\[\mu_{m\widehat{\sigma}(m)}\prod_{i=1}^{i_0}\mu_{\widehat{\sigma}^{-1}(\delta_i)\delta_i}\leq\mu_{mm},\]

Therefore
\begin{align*}
\alpha^{\widehat{\beta}} & = \prod_{i\in I} \mu_{i\widehat{\sigma}(i)} \prod_{i\not\in I} \mu_{i\widehat{\sigma}(i)} \\
& \leq  \mu_{mm} \prod_{i\not\in I\cup\{m\}} \mu_{i\widehat{\sigma}(i)} \\
& <  \mu_{m\widetilde{\sigma}(m)} \prod_{i\not\in I\cup\{m\}} \mu_{i\widetilde{\sigma}(i)} \prod_{i\in I} \mu_{i\widetilde{\sigma}(i)}
\end{align*}

which implies that $\widehat{\beta}$ is not a maximum, that is $\widehat{\beta}<\beta$.
\\

\underline{\textbf{Case 2:}} If $l_i=k_i$ for some $i \in \{ 1,\dots,h \}$, then $\overline{\sigma}(m)=l_{\overline{i}}$, where $\overline{i}$ is the minimum $i\in\{1,\dots, h\}$ such that $l_i=k_i$.
\\

\emph{Proof Case 2:}  Take $\widehat{\sigma}\in S_m$ with
$$
\mu_{m \widehat{\sigma}(m)}\mu_{m-1 \widehat{\sigma}(m-1)}\cdots \mu_{1 \widehat{\sigma}(1)}\neq 0,
$$
and $\mu_{m \widehat{\sigma}(m)} \neq \mu_{m l_{\overline{i}}}$. Let
\[\alpha^{\widehat{\beta}}= \mu_{m \widehat{\sigma}(m)}\mu_{m-1 \widehat{\sigma}(m-1)}\cdots \mu_{1 \widehat{\sigma}(1)} .
\]

It is clear that in this case $\widehat{\sigma}(m) \widehat{\sigma}(m-1) \cdots \widehat{\sigma}(m- l_{\overline{i}} +1)$ necessarily belong to the set $\{ 1,\dots, l_{\overline{i}}   \}$. Hence, one can consider the matrix $A'$ form by the first $l_{\overline{i}}$ columns and the last $l_{\overline{i}}$ rows of $A$. Applying the previous reasoning it is straightforward to see that we are now in the situation of the case 1 for the new matrix $A'$. Thus, $\mu_{m\overline{\sigma}(m)}= \mu_{m l_{\overline{i}}}$.
\\

Once $\overline{\sigma}(m)$ has been uniquely determined, we can remove from $A$ its $m$-th row and its $\overline{\sigma}(m)$-th column to obtain a new square matrix $A_1$ of order $m-1$. We follow the same previous arguments applied to $A_1$ instead of to $A$ to determine $\overline{\sigma}(m-1)$. In this way we can uniquely determine $\overline{\sigma}\in S_m$ and therefore prove the existence of a unique maximum in the terms of (\ref{deter}).

Next, we prove the bound on the size of a field $\F$ in order to ${\cal T}( T_0, \dots  , T_{L})$ to be superregular over $\F$.

We just proved that there exists a unique term with highest exponent in each nontrivial determinant of every submatrix $A=[\mu_{ij}]$ of ${\cal T}( T_0,T_1, \dots  , T_{L})$. Let
\begin{equation}\label{eq_02}
\alpha^{\beta} =\mu_{(r) \overline{\sigma}(r)}\mu_{(r-1) \overline{\sigma}(r-1)}\cdots \mu_{1 \overline{\sigma}(1)}
\end{equation}
for some $\overline{\sigma}\in S_{r}$ with $2 \leq r \leq (L+1)M$, be the highest term one can find.
Define $R_k=\{(i,j)\in \N^2\ | \ 1\leq i,j \leq k \mbox{ and } i=k \mbox{ or } j=k   \}$ and $R(\overline{\sigma})=\{ (i,j)\in \N^2\ | \ (i,j)=(t,\overline{\sigma}(t)) \mbox{ for some } t\in \{ 1,2,\dots, r \} \}$. It follows from the properties of the matrix ${\cal T}( T_0,T_1, \dots  , T_{L})$ that
$$
\displaystyle \prod_{(i,j)\in R_k \cap R(\overline{\sigma})} \mu_{ij} \leq \alpha^{2^{M(L+2)- 2(r-k+1)}}
$$
for  $k=1,\dots,r$. Hence,
\begin{eqnarray*}
\alpha^{\beta} &=& \displaystyle \prod_{k=1}^r \prod_{(i,j)\in R_k \cap R(\overline{\sigma})} \mu_{ij}  \leq \alpha^{2^{M(L+2)-2}} \alpha^{2^{M(L+2)-4}} \cdots \alpha^{2^{M(L+2)-2r}}  < \alpha^{(2^{M(L+2)-2 } )( \sum_{i=0}^{\infty} 4^{-i})}\\
   &=& \alpha^{(2^{M(L+2)-2})( \frac{4}{3})} < \alpha^{2^{M(L+2)-1}} .
\end{eqnarray*}
So $\beta < 2^{m(L+2)-1}$.
This means that the maximum exponent of $\alpha$ appearing in the determinants of the submatrices of ${\cal T}( T_0,T_1, \dots  , T_{L})$ is upper bounded by $2^{M(L+2)-1}$.

Notice that deg$(\rho(z))\geq 2^{M(L+2)-1}$. If ${\cal T}( T_0,T_1, \dots  , T_{L})$ is not superregular over $\F$ then there exists a nontrivial determinant $f(\alpha)= \displaystyle \sum_{i=0}^{2^{M(L+2)-1}} \epsilon_i \alpha_i$, $\epsilon_i \in \{0,1,2, \dots, p-1\}$ of a submatrix $A$ of ${\cal T}( T_0,T_1, \dots  , T_{L})$ such that $f(\alpha)=0$. By Theorem \ref{th:classic_finite_fields} it follows that $\rho(z)|f(z)$ which contradicts the fact that the degree of $f(z)$ is less than $2^{M(L+2)-1}$.\end{proof}

\begin{remark}
Note that if $A$ is a submatrix of a matrix $B$ of the form (\ref{eq:04}) or (\ref{eq:07}) then $|A|$ is a trivial minor of $B$ if it contains zeros in its diagonal, and therefore in order to check the superregularity of $B$ it is enough to verify the determinant of the submatrices with no zero elements in its diagonal. More concretely, it can be checked that if $A=[\mu_{ij}]$ is a square submatrix of order $r$ of a matrix of the form (\ref{eq:04}) or (\ref{eq:07}) then $|A|$ is a nontrivial determinant if and only if its indices $i_1< i_2 < \cdots < i_r \leq (j+1)(n-k)$ and $j_1 < \cdots < j_r \leq (j+1)k$ satisfies $j_t \leq \left \lfloor \dfrac{i_t}{n-k} \right \rfloor k$ for $t=1, \dots, r$.
\end{remark}

It is well-known that if $N$ is an integer and $p$ a prime number then there exists a finite field $\F$ with $p^N$ elements and therefore there exists a finite field $\F$ such that $|\F|=p^{(2^{M(L+2)-1})}$. However, it follows from the proof of Theorem \ref{th:mainFieldSize} that it is enough to have $|\F|>p^{((2^{M(L+2)-2})( \frac{4}{3}))}$ in order to ${\cal T}( T_0,T_1, \dots  , T_{L})$ to be superregular. It can be checked using computer algebra programs that there are particular examples (for small values of $(n,k,\delta)$) of superregular matrices that require a much smaller field size, see for instance \cite[Example 3.10]{Gluesing-Luerssen2006}. However, the proposed superregular matrices can be constructed for any given characteristic $p$ and parameters $(n,k,\delta)$ and therefore provides a general construction. Note that the superregular matrix $S_r$ given in (\ref{eq_06}) requires, in general, a large characteristic $p(r)$.

We are now in the position to present a new class of MDP convolutional codes. The result easily follows from Theorem \ref{th:MDP_char}, Theorem \ref{th:mainFieldSize} and the fact that submatrices of a superregular matrix inherit the superregularity property.

\begin{corollary} \label{corfinal}
Let $(n,k,\delta)$ be given and let $T_\ell=[t_{ij}^\ell]$, $i,j=1,2,\dots, m$ and $\ell=0,1,2\dots,L$ be the entries of the matrix $T_\ell$ as in (\ref{eq_01}). Define ${\bar H}_\ell=[t_{ij}^\ell]$   $i=1,2,\dots, (n-k)$, $j=1,2,\dots, k$ and $\ell=0,1,2\dots,L$. If $| \mathbb F | \geq p^{(2^{M(L+1)+n-2})}$ then, the convolutional code ${\cal C}=\ker_{\mathbb F[z]}  [A(z) \ B(z)] $ where $A(z)=\displaystyle \sum_{i=0}^{\nu} A_i z^i\in \mathbb F[z]^{(n-k)\times (n-k)}$ and $B(z)=\displaystyle \sum_{i=0}^{\nu} B_i z^i\in \mathbb F[z]^{(n-k)\times k}$, with $\nu = \frac{\delta}{n-k}$, $A_0=I_{n-k}$, $A_i \in \mathbb F^{(n-k)\times (n-k)}$, $i = 1, \dots, \nu$ obtained by solving the equations
$$
[A_{\nu} \cdots A_1]\left[\begin{matrix} {\bar H}_{L- \nu} & \cdots & {\bar H}_1 \\
{\bar H}_{L- \nu+1} & \cdots & {\bar H}_2 \\ \vdots & & \vdots \\ {\bar H}_{L-1} & \cdots & {\bar H}_{\nu} \\  \end{matrix} \right] = -[{\bar H}_L \cdots {\bar H}_{\nu+1}],
$$
and $B_i=A_0{\bar H}_i + A_1 {\bar H}_{i-1} + \cdots + A_i {\bar H}_0, i=0, \dots, \nu,
$
is an MDP convolutional code of rate $k/n$ and degree $\delta$.
\end{corollary}

\begin{remark}
Details about the construction of the matrices $A(z)$ and $B(z)$ presented in Corollary \ref{corfinal} can be found in \cite[Appendix C]{Gluesing-Luerssen2006}
\end{remark}

The following example illustrates the construction of a $(5,2,3)$ MDP convolutional code.

\begin{example}
Since $n=5$, $k=2$ and $\delta = 3$, we have that $L=2$ and $\nu=1$.
Let us consider $\alpha$ a root of the primitive polynomial $x^{1024} + x^{39} + x^{37} + x^{36} + 1 \in \mathbb F_2[x]$, i.e., a primitive element over the field $\mathbb F_{2^{1024}}$ and the matrix
$$
[{\bar H}_0 \ {\bar H}_1 \ {\bar H}_2]=\left[\begin{matrix}
\alpha^{2^0} & \alpha^{2^1} & | & \alpha^{2^3} & \alpha^{2^4} & | & \alpha^{2^6} & \alpha^{2^7}\\
\alpha^{2^1} & \alpha^{2^2} &| & \alpha^{2^4} & \alpha^{2^5} & | & \alpha^{2^7} & \alpha^{2^8}  \\
\alpha^{2^2} & \alpha^{2^3} & | & \alpha^{2^5} & \alpha^{2^6} & | &  \alpha^{2^8} & \alpha^{2^9}
\end{matrix}\right]
$$
over $\mathbb F_{2^{1024}}$. Considering $A(z) = I_3 + A_1 z$ such that $A_1 {\bar H}_1 = -{\bar H}_2$, where, a possible choice is
\begin{eqnarray*}
A(z)= & \\
\left[ \begin{matrix} 1 & 0 & 0 \\ 0 & 1 & 0 \\ 0 & 0 & 1  \end{matrix}\right] + &\frac{1}{\alpha^{2^3+2^5}-\alpha^{2^5}} \left[ \begin{matrix} -\alpha^{2^5+2^6} + \alpha^{2^4+2^7} & - \alpha^{2^5+2^7} + \alpha^{2^4+2^8} & -\alpha^{2^5+2^8} + \alpha^{2^4+2^9}\\ \alpha^{2^4+2^6}- \alpha^{2^3+2^7} & \alpha^{2^4+2^7} - \alpha^{2^3+2^8} & \alpha^{2^4+2^8} - \alpha^{2^4+2^8} - \alpha^{2^3+2^9} \\ 0 & 0 & 0  \end{matrix}\right]z,
\end{eqnarray*}
and $B(z)=B_0+ B_1 z$ such that $B_0 = {\bar H}_0$ and $B_1 = {\bar H}_1 + A_1 {\bar H}_0$,
we have that
$$ {\cal C} = \ker_{\mathbb F[z]}  [A(z) \; B(z)]
$$
is a $(5,2,3)$ MDP convolutional code.
\end{example}

\section{Conclusions}

There is a type of superregular matrices that are essential for the construction of MDP convolutional codes. However, very little is understood about how to construct these matrices and how large a finite field must be, so that a superregular matrix of a given order can exist over that field. In this paper, we have presented a new class of MDP $(n, k, \delta)$ convolutional codes, such that $(n-k) | \delta$ , by means of the construction of a novel type of superregular matrices over a field of any characteristic. We also
established a bound for the size of the field needed for these matrices to be superregular. 


\end{document}